\documentclass[conference]{IEEEtran}
\IEEEoverridecommandlockouts

\usepackage{cite}
\usepackage{amsmath,amssymb,amsfonts}
\usepackage{amsthm}
\usepackage{algorithmic}
\usepackage{graphicx}
 
\usepackage{eso-pic}
 \usepackage{pgfplots}
 \usepgfplotslibrary{groupplots}
 \pgfplotsset{compat=1.18}

\usepackage{textcomp}
\usepackage{xcolor}
\def\BibTeX{{\rm B\kern-.05em{\sc i\kern-.025em b}\kern-.08em
    T\kern-.1667em\lower.7ex\hbox{E}\kern-.125emX}}
\usepackage{mathtools}
\usepackage{bm}
\usepackage{booktabs}
\usepackage{array}
\usepackage{multirow}
\usepackage{algorithm}
\usepackage{hyperref}
\hypersetup{colorlinks=true,linkcolor=blue,citecolor=blue,urlcolor=blue}

\usepackage{comment}
\newcommand{\R}{\mathbb{R}}

\newcommand{\norm}[1]{\left\lVert#1\right\rVert}
\newcommand{\image}{\operatorname{image}}
\newcommand{\rank}{\operatorname{rank}}
\newcommand{\Hank}{\mathcal{H}}
\newcommand{\col}{\operatorname{col}}

\newcommand{\BL}{\mathcal{B}_{L}}
\newcommand{\BS}[1]{\mathcal{B}_{L}^{(#1)}}

\theoremstyle{plain}
\newtheorem{theorem}{\bf Theorem}

\newtheorem{proposition}[theorem]{\bf Proposition}
\newtheorem{corollary}[theorem]{\bf Corollary}
\theoremstyle{definition}
\newtheorem{definition}{\bf Definition}
\newtheorem{assumption}{\bf Assumption}
\newtheorem{remark}{\bf Remark}

\begin{document}

\title{Iterative Behavioral Identification for\\
Covert Misappropriation Attacks%
\thanks{This work was  partially supported  by CAPES, Finance Code 001; by FAPERJ under grant E-26/204.268/2024; and by the Conselho Nacional de Desenvolvimento Científico e Tecnológico (CNPq), Brazil, under grants 403601/2023-1,  315106/2023-9, and BPP/PQ-Sr  313335/2022-2, as well as  FCT through the LASIGE Research Unit, ref. UID/00408/2025, DOI https://doi.org/10.54499/UID/00408/2025.}}

\author{%
\IEEEauthorblockN{%
Victor Mattos\IEEEauthorrefmark{1},
Amit Bhaya\IEEEauthorrefmark{1},
Daniel Sadoc Menasche\IEEEauthorrefmark{4},
Alan O.~de S\'{a}\IEEEauthorrefmark{3}}
\IEEEauthorblockA{%
\IEEEauthorrefmark{1}Dept.\ Electrical Engineering, UFRJ,
Rio de Janeiro, Brazil.
\texttt{victor.mattos@coppe.ufrj.br}}
\IEEEauthorblockA{%
\IEEEauthorrefmark{4}Institute of Computing, UFRJ, Rio de Janeiro, Brazil.}
\IEEEauthorblockA{%
\IEEEauthorrefmark{3}LASIGE, Faculdade de Ciências, Universidade de Lisboa, Lisbon, Portugal.}}

\AddToShipoutPictureFG*{%
  \AtPageLowerLeft{%
    \put(\LenToUnit{0.75in},\LenToUnit{0.55in}){%
      \parbox[b]{7.0in}{%
        \fontsize{7}{8}\selectfont
        \textcopyright{} 2026 IEEE.
        Personal use of this material is permitted.
        Permission from IEEE must be obtained for all
        other uses, in any current or future media,
        including reprinting/republishing this material
        for advertising or promotional purposes,
        creating new collective works, for resale or
        redistribution to servers or lists, or reuse
        of any copyrighted component of this work
        in other works.
      }%
    }%
  }%
}

\maketitle
\enlargethispage{0.05\baselineskip}
\begin{abstract}
		We present a data-driven, three-phase framework for covert misappropriation attacks on closed-loop networked control systems. A man-in-the-middle (MITM) adversary, with access only to intercepted input–output data, incrementally identifies the plant’s behavioral subspace while maintaining stealth under an infinity-norm intrusion-detection system (IDS). Phase~1 collects an initial rank-deficient Hankel matrix using a low-amplitude probing signal that satisfies an analytical stealth bound. Phase~2 adapts online rank-increasing experiment design to the closed-loop setting, using causal predictor-based masking to limit IDS residuals while the Hankel rank grows from $1$
		to the full behavioral rank $n+mL$ in exactly $n+mL-1$ steps.
		We show that the identified subspaces are strictly nested and that the projection error of any desired trajectory decreases monotonically, strictly so almost surely at each rank increase. Phase~3 employs the final Hankel matrix in a DeePC-style optimization to synthesize a covert attack that tracks an arbitrary reference while respecting the IDS constraint. Simulations on a DC motor with PI control demonstrate full-rank identification and stealthy tracking. Stealth is guaranteed analytically in Phases~1 and~3 and is numerically supported in Phase~2. The method is model-free and assumes noise-free LTI dynamics.
\end{abstract}


	\section{Introduction}
	\label{sec:intro}
	
	Cyber-physical systems increasingly rely on networked communication, making them vulnerable to adversaries capable of manipulating sensor and actuator signals. Among these threats, \emph{covert attacks} are particularly dangerous: an attacker steers the physical process toward an unintended trajectory while keeping observable quantities within limits that avoid triggering alarms. Their stealthy nature makes such attacks difficult to detect and potentially severe.
	
	We consider a \emph{man-in-the-middle} (MITM) adversary positioned between the plant and controller. The attacker observes and alters both channels, but begins with \emph{no} knowledge of the plant or controller model. The closed-loop system is assumed linear time-invariant (LTI), a standard approximation for industrial processes near a setpoint. The attacker’s objective is to \emph{learn} the closed-loop input–output behavior from intercepted data and subsequently drive the plant along a malicious reference while remaining undetected.
	
	This setting creates a fundamental conflict:
	\begin{itemize}
		\setlength{\itemsep}{0pt}\setlength{\parskip}{0pt}
		\item[(T1)] \emph{Information need}: informative, persistently exciting (PE) data are required to build a useful predictor;
		\item[(T2)] \emph{Stealth need}: large injections produce output deviations that violate the IDS threshold.
	\end{itemize}
	Our framework resolves this conflict through a structured, incremental design.
	
	Three components make this possible: the \emph{behavioral framework}~\cite{willems2005}, representing plant trajectories as linear combinations of Hankel matrix columns for model-free prediction; the \emph{online experiment design} of~\cite{waardePersistentExcitationOnline2021}, increasing Hankel rank one column at a time; and the \emph{covert misappropriation} architecture~\cite{Smith2015}, supplying a masking mechanism for stealth.

These ingredients yield a three-phase attack pipeline:
\begin{itemize}
    \setlength{\itemsep}{0pt}\setlength{\parskip}{0pt}
    \item \textbf{Phase~1} (initial covert probing): a low-amplitude
Rademacher injection produces an initial rank-deficient Hankel matrix under an analytical stealth bound.
\item \textbf{Phase~2} (iterative covert identification): each iteration
injects a rank-increasing signal and applies a predictor-based masking. The attacker's representational subspace expands strictly at each rank increase, and the projection error of the desired attack trajectory onto
this subspace decreases monotonically (Theorem~\ref{thm:strict-decrease}).

\item \textbf{Phase~3} (behavioral misappropriation): the final Hankel
matrix is used in a DeePC optimization~\cite{Coulson2019} to drive
the plant to an attacker-chosen trajectory satisfying the IDS
constraint.
\end{itemize}
	
	Stealth is guaranteed analytically in Phases~1 and~3, whereas in
	Phase~2 it is supported numerically.

    The remainder of the paper follows the structure outlined above: Section~\ref{sec:context} describes the problem context and contributions; Section~\ref{sec:model} presents the system model; Section~\ref{sec:threat} defines the threat model and IDS. Sections~\ref{sec:phase1}–\ref{sec:phase2} detail Phases~1 and~2. Section~\ref{sec:theory} establishes theoretical properties; Section~\ref{sec:phase3} presents Phase~3. Section~\ref{sec:sims} reports simulations. Section~\ref{sec:discussion} discusses assumptions and Section~\ref{sec:conc} concludes.

\section{Problem Context, Literature, and Contributions}
	\label{sec:context}
	
	\subsection{Networked Control Systems under MITM Attacks}
	
	Networked control systems in industrial automation are vulnerable to
	MITM adversaries that intercept and alter controller–actuator and
	sensor–controller signals in real time, realistic in SCADA and ICS
	environments, where routers and protocol converters may be breached.
	The attacker’s objective is \emph{covert misappropriation}: steering
	the plant toward an alternative setpoint or trajectory while keeping
	all monitored quantities within the IDS limits so alarms remain silent.
	
	\subsection{Limitations of Existing Work}
	
	\paragraph*{Model-based covert attacks}
	Classical covert-attack designs, including covert misappropriation~\cite{Smith2015},
	zero-dynamics attacks~\cite{teixeira2015}, replay attacks~\cite{mo2009}, and
	false-data injection~\cite{liu2011}, assume the adversary knows the plant
	and controller dynamics or relevant closed-loop transfer functions,
	unrealistic for a network-layer attacker.

	\paragraph*{Passive data-driven identification}
	Purely passive identification is
	insufficient under steady-state regulation, where closed-loop trajectories are
	typically far from persistently exciting (PE): without active probing the
	attacker cannot obtain a rich predictor, yet active probing without stealth
	constraints is easily detected.
	
	These observations reveal three gaps:
	\begin{itemize}
		\setlength{\itemsep}{0pt}\setlength{\parskip}{0pt}
		\item[(L1)] \textbf{Model knowledge}: existing attacks require plant or
		controller models unavailable to a MITM adversary.
		\item[(L2)] \textbf{Passive ID}: passive observation does not yield PE data,
		while naive active identification violates stealth.
		\item[(L3)] \textbf{Closed-loop behavioral synthesis}: no prior work combines
		model-free behavioral prediction~\cite{willems2005,Coulson2019} with
		stealth-constrained active identification and closed-loop attack synthesis.
	\end{itemize}
	
	Recent efforts address isolated aspects of these gaps~\cite{duStealthinessAnalysisIdentified2024, bhowmik2025}, but none provide a unified, stealth-constrained framework.
	
	\subsection{Contributions}
	
	This paper closes gaps (L1)–(L3) through a three-phase, incremental,
	model-free attack pipeline operating solely on intercepted closed-loop data:
	
	\begin{enumerate}
		\setlength{\itemsep}{0pt}\setlength{\parskip}{0pt}
		\item \textbf{Three-phase covert framework}: a complete pipeline requiring no
		a priori model knowledge and integrating probing, identification, and
		behavioral attack synthesis.
		
		\item \textbf{Closed-loop sample complexity}: Remark~\ref{rmk:closed-loop}
		shows that the attacker’s additive authority over the actuator channel makes
		the closed-loop rank-increasing procedure equivalent to the open-loop design
		of~\cite{waardePersistentExcitationOnline2021}. Full behavioral rank is
		achieved with $T_1=L$ Phase-1 samples and $T_2=n+mL-1$ Phase-2 steps, for a
		total of $T^\star=n+(m+1)L-1$ samples.
		
		\item \textbf{Nested subspaces and monotone projection error}: Theorem
		\ref{thm:strict-decrease} and Corollary~\ref{cor:finite-reconstruction}
		establish that the identified subspaces are strictly nested and that the
		projection error $d_k(w^\star)$ decreases monotonically, strictly almost
		surely at each rank increase.
		
		\item \textbf{DeePC-based covert attack synthesis}: the final Hankel matrix is
		deployed in a convex DeePC-style QP (Section~\ref{sec:phase3}) to generate a
		covert misappropriation attack satisfying explicit IDS constraints.
	\end{enumerate}
	
	All analysis assumes noise-free LTI dynamics; robustness to noise is left for
	future work.

\section{System Model and Behavioral Preliminaries}
\label{sec:model}

\subsection{Plant, Controller, and Closed Loop}

The plant is a discrete-time LTI system of \emph{unknown} order~$n$:
\begin{equation}
  x_{t+1} = A x_t + B u_t,\qquad
  y_t = C x_t,
  \label{eq:plant}
\end{equation}
with state $x_t \in \R^n$, input $u_t \in \R^m$, and output
$y_t \in \R^p$.
The matrices $(A,B,C)$, the lag $\ell$ and the exact order $n$ are \emph{unknown to the
attacker}.

The plant is regulated by an LTI controller,
\begin{equation}
  x_{c,t+1} = A_c x_{c,t} + B_c \tilde y_t,\qquad
  u_{c,t}   = C_c x_{c,t} + D_c \tilde y_t,
  \label{eq:controller}
\end{equation}
where $\tilde y_t$ is the sensor signal received by the controller, possibly modified by the attacker. $x_{c,t} \in \R^{n_c}$ is
the controller state, $(A_c,B_c,C_c,D_c)$ and $n_c$ are \emph{unknown}.

\begin{assumption}[Stability and minimality]
\label{ass:stability}
The plant~\eqref{eq:plant} is controllable and observable, and the controller~\eqref{eq:controller} stabilizes the closed loop.
\end{assumption}

\begin{assumption}[Steady-state operation]
\label{ass:steadystate}
In nominal operation the closed loop remains near a constant setpoint.
From a pre-attack window, the attacker estimates $(\bar u, \bar y)$ and bounds $\varepsilon_{\mathrm{ss}},
\varepsilon_u \ge 0$ such that
$\norm{y_t - \bar y}_\infty \le \varepsilon_{\mathrm{ss}}$ and
$\norm{u_{c,t} - \bar u}_\infty \le \varepsilon_u$ for all $t$.
\end{assumption}

\subsection{Peak Gain of the Actuator-to-Output Map}

Stealth constraints depend on the closed-loop peak ($\ell^\infty$) gain.

\begin{definition}[Empirical peak gain]
\label{def:peak-gain}
For an additive actuator perturbation
$\Delta u_{[0,t]}$ to the resulting output deviation
$\Delta y_{[0,t]}$, the
\emph{peak gain} is
\begin{equation}
  \hat G_\infty :=
  \sup_{\Delta u \neq 0}
  \frac{\norm{\Delta y}_\infty}{\norm{\Delta u}_\infty},
  \label{eq:peak-gain}
\end{equation}
estimated or upper-bounded from a pre-attack probing.

\end{definition}

An injection satisfying
$\norm{\Delta u}_\infty \le
(\varepsilon_{\mathrm{ids}}-\varepsilon_{\mathrm{ss}})/\hat G_\infty$
keeps the \emph{unmasked} deviation
$\norm{\tilde y_t-\bar y}_\infty$ within the IDS margin, since the
nominal output already occupies up to $\varepsilon_{\mathrm{ss}}$ of
the budget (Assumption~\ref{ass:steadystate}); the margin is positive
because $\varepsilon_{\mathrm{ids}}>\varepsilon_{\mathrm{ss}}$
(Definition~\ref{def:ids} below).
Masking in Phase 2 can further reduce the realized residual.

\subsection{Behavioral Representation and Hankel Matrices}

We adopt the behavioral viewpoint~\cite{willems2005}.
Let $q = m+p$ and $w_t = \col(u_t,y_t) \in \R^q$.
The \emph{length-$L$ behavior} of the plant is
\begin{equation}
  \label{eq:behavior}
  \begin{aligned}
    \BL := \bigl\{\, &(w_0,\dots,w_{L-1})\in(\R^q)^L \bigm| \\
    &w \text{ is a length-$L$ trajectory of \eqref{eq:plant}} \,\bigr\}.
  \end{aligned}
\end{equation}
Although the data are
gathered \emph{inside} the closed loop, each recorded pair
$(u_t,y_t)$ is a plant
trajectory, so the collected Hankel windows sample $\BL$.
Passive data typically span only a low-dimensional subset; active
injection is required to explore $\BL$.

For a sequence $\{w_t\}$ of length-$T$ and depth $\ell \le L \le T$, the \emph{depth-$L$ Hankel matrix} is
\begin{equation}
  \Hank_L(w) :=
  \begin{bmatrix}
    w_0    & w_1    & \cdots & w_{T-L}   \\
    w_1    & w_2    & \cdots & w_{T-L+1} \\
    \vdots & \vdots & \ddots & \vdots    \\
    w_{L-1}& w_L   & \cdots & w_{T-1}
  \end{bmatrix}
  \in \R^{qL\times(T-L+1)}.
  \label{eq:hankel}
\end{equation}

\begin{theorem}[Fundamental Lemma~\cite{willems2005}]
\label{thm:fundamental}
If the plant is controllable and observable of order
$n$ and $u^d$ is persistently exciting of order $L+n$, then
\begin{equation}
  \rank\Hank_L(w^d) = n + mL, \quad \image\Hank_L(w^d) = \BL .
  \label{eq:fl}
\end{equation}
\end{theorem}

\section{Threat Model and IDS}
\label{sec:threat}

\subsection{MITM Channels}

 A MITM adversary intercepts and modifies both actuator and sensor channels:
\begin{align}
  u_t &= u_{c,t} + u_{a,t}
  \label{eq:attack-actuator}\\[-1mm]
  \tilde y_t &= y_t + y_{a,t},
  \label{eq:attack-sensor}
\end{align}
where $u_{c,t}$ is the controller output and $y_t$ the true plant
output. The attacker injects 
 $u_{a,t}$ 
and masks the sensor signal via $y_{a,t}$.
The controller and IDS both receive $\tilde y_t$.

\subsection{Attacker Knowledge}

\textbf{Known}: IDS threshold $\varepsilon_{\mathrm{ids}}$;
steady-state estimates $(\bar u, \bar y)$; peak-gain estimate
$\hat G_\infty$ ;
that the plant and controller are LTI and closed-loop stable;
full intercept history.  

\textbf{Unknown}: plant and controller parameters, orders,
and all closed-loop or sensitivity transfer functions.

\subsection{IDS Model}

\begin{definition}[Infinity-norm IDS]
\label{def:ids}
The IDS triggers at time $t$ if
$\norm{\tilde y_t - \bar y}_\infty > \varepsilon_{\mathrm{ids}}$,
with $\varepsilon_{\mathrm{ids}} > \varepsilon_{\mathrm{ss}} \ge 0$.
An attack is \emph{stealthy} if the IDS never triggers.
\end{definition}

The IDS monitors only $\tilde y_t$; actuator injections are unobserved,
making masking the operative stealth mechanism. The $\ell^\infty$ form
yields simple analytical bounds.

\section{Phase~1: Initial Covert Probing}
\label{sec:phase1}

\subsection{Injection Design}

During $t \in [0,T_1-1]$, the attacker injects a low-amplitude Rademacher signal:
\begin{equation}
  u_{a,t} = \delta_1 s_t,\quad s_t \in \{-1,+1\}\
  \text{i.i.d.},\quad y_{a,t} = 0,
  \label{eq:phase1-injection}
\end{equation}
with no masking applied. Rademacher probing is simple to implement, spectrally flat, and widely used in online identification.

\begin{assumption}[Phase~1 stealth]
\label{ass:phase1-stealth}
The amplitude $\delta_1>0$ satisfies
$\delta_1 \le
(\varepsilon_{\mathrm{ids}}-\varepsilon_{\mathrm{ss}})/\hat G_\infty$,
so that
$\norm{y_t-\bar y}_\infty
\le \varepsilon_{\mathrm{ss}}+\hat G_\infty\delta_1
\le \varepsilon_{\mathrm{ids}}$
for all $t\in[0,T_1-1]$ (cf.\ Definition~\ref{def:peak-gain}).
\end{assumption}

\subsection{Initial Hankel Matrix and the Role of $T_1=L$}

From the recorded data $\{(u_t,y_t)\}_{t=0}^{T_1-1}$ the attacker constructs
\begin{equation}
  H^{(0)} := \Hank_L(w),\quad
  w_t = \col(u_t,y_t),\quad
  L := T_{\mathrm{ini}} + N,
  \label{eq:H0}
\end{equation}
where $T_{\mathrm{ini}}$ and $N$ are the past and future DeePC
horizons.

Following~\cite{waardePersistentExcitationOnline2021},
the first $L$ input samples may be arbitrary,
provided they are not all zero.
Choosing $T_1=L$ then yields a single nonzero column
in $H^{(0)}$, so $\BS{0}:=\image H^{(0)}$
is one-dimensional, i.e.\
$\rank H^{(0)}=1\ll n+mL$, and Phase~2 is necessary.

\begin{remark}[Why not skip Phase~2?]
\label{rem:why-phase2}
A rank-$1$ predictor can reproduce only scalar multiples of the single
observed window.
A generic target $w^\star$ has a large component orthogonal to that
one direction, so the projection error $d_0(w^\star)$ is on the order
of $\norm{w^\star}$.
Used directly in Phase~3, such a predictor yields poor tracking and
large IDS residuals.
Phase~2 enlarges the span precisely to drive this error to zero.
\end{remark}

\section{Phase~2: Iterative Covert Behavioral Identification}
\label{sec:phase2}

\subsection{Closed-Loop Sample Complexity}

In the open-loop setting, prior work~\cite{waardePersistentExcitationOnline2021}
shows that, by choosing the
input at each step so that the new column of $\Hank_L(w)$ leaves the
current image, full behavioral rank is reached after exactly
\begin{equation}
  T^\star = n + (m+1)L - 1 \quad \text{samples.}
  \label{eq:T-star}
\end{equation}
This online design is sample-economical: the offline Fundamental Lemma
requires at least
$(m+1)(L+n)-1$ samples, so~\eqref{eq:T-star} saves
$(m+1)(L+n)-1-T^\star = mn$ samples.

The open-loop design assumes the designer chooses $u_t$ freely.
The following observation shows that the MITM adversary enjoys exactly this
freedom, so the bound transfers without loss.

\begin{remark}[Closed-loop sample complexity via control authority]
\label{rmk:closed-loop}
Under Assumption~\ref{ass:stability}, consider the actuator
channel~\eqref{eq:attack-actuator} with $u_{c,t}$ intercepted before
injection.
Because $u_{a,t}\in\R^m$ is unconstrained in direction, the realized input
$u_t=u_{c,t}+u_{a,t}$ can be assigned any direction in $\R^m$ at each step.
Hence the attacker can implement any open-loop input sequence, in
particular the online rank-increasing design
of~\cite{waardePersistentExcitationOnline2021}, and attains
$\rank H^{(K)} = n+mL$, i.e.\ $\BS{K}=\BL$, in exactly
\begin{equation}
  T_2 := K = T^\star - L = n+mL-1
  \label{eq:T-two}
\end{equation}
rank-increasing steps.
Together with the $T_1=L$ samples of Phase~1, the total experiment
length is $T_1+T_2 = T^\star = n+(m+1)L-1$, matching the open-loop
sample complexity.

\end{remark}

\subsection{Stealth-Friendly Rank-Increasing Injection}
\label{subsec:rank-check}

We now give the additive realization that increases rank while
perturbing the loop as little as possible.
At time $t$, after observing $(u_{c,t},y_t)$, the attacker decides
whether to inject.
The upcoming column of $\Hank_L(w)$ is \emph{redundant} iff there is a
left-kernel vector $\xi=(\xi_1,\dots,\xi_L)\in(\R^q)^L$ of $H^{(k)}$,
not all zero, with
\begin{equation}
  \sum_{i=1}^{L} \xi_i^\top w_{t-L+i} = 0
  \quad \text{for every previously observed window.}
  \label{eq:kernel-condition}
\end{equation}
Partition $\xi_i = \col(\xi_i^u,\xi_i^y)$ conformably with
$w_t = \col(u_t,y_t)$.
Substituting $u_t = u_{c,t} + u_{a,t}$ isolates the injection:
\begin{equation}
  (\xi_L^u)^\top u_{a,t} = \alpha_t,
  \label{eq:ua-balance}
\end{equation}
where
\begin{equation}
\begin{split}
  \alpha_t :=\;&
  -\!\sum_{i=1}^{L-1}\bigl[(\xi_i^y)^\top y_{t-L+i}
  +(\xi_i^u)^\top u_{t-L+i}\bigr]\\
  &-(\xi_L^y)^\top y_t
  -(\xi_L^u)^\top u_{c,t},
\end{split}
  \label{eq:alpha-def}
\end{equation}
and every quantity in~\eqref{eq:alpha-def} is known to the attacker.

\begin{proposition}[Rank-increasing condition in closed loop]
\label{prop:rank-increase}
Suppose~\eqref{eq:kernel-condition} holds at time $t$ with nontrivial
$\{\xi_i\}$ and $\xi_L^u \neq 0$.
Then any injection satisfying
\begin{equation}
  (\xi_L^u)^\top u_{a,t} \neq \alpha_t
  \label{eq:ua-condition}
\end{equation}
makes the new column of $\Hank_L(w)$ linearly independent of all
previous columns, increasing $\rank H^{(k)}$ by exactly one.
\end{proposition}

A simple stealth-friendly choice satisfying~\eqref{eq:ua-condition} is
\begin{equation}
  u_{a,t} = \delta_k\,\frac{\xi_L^u}{\norm{\xi_L^u}_2},
  \qquad
  \delta_k \neq \frac{\alpha_t}{\norm{\xi_L^u}_2},
  \label{eq:ua-choice}
\end{equation}
which violates~\eqref{eq:ua-balance} and hence increases the rank by
Proposition~\ref{prop:rank-increase}. The amplitude $|\delta_k|$ is
\emph{not} constrained by the unmasked budget
$(\varepsilon_{\mathrm{ids}}-\varepsilon_{\mathrm{ss}})/\hat G_\infty$
of Assumption~\ref{ass:phase1-stealth}: that bound governs Phase~1,
where no masking is applied. In Phase~2 the operative stealth
constraint is instead on the \emph{masked} residual
$\tilde y_t-\bar y$, which the sensor masking~\eqref{eq:masking}
drives toward zero as $H^{(k)}$ improves
(Remark~\ref{rem:virtuous-cycle}).

\begin{assumption}[Closed-loop non-degeneracy]
\label{ass:nondegenerate-kernel}
At every Phase~2 iteration $k$ with $\rank H^{(k)} < n+mL$, there is a
nontrivial left-kernel vector $\{\xi_i\}_{i=1}^L$ of $H^{(k)}$ whose
last input block is nonzero, $\xi_L^u\neq 0$.
\end{assumption}

\subsection{Sensor Masking During Phase~2}

At each step the attacker masks the induced deviation.
Let $\hat y_t$ be the output predicted by the current behavioral
predictor (built from $H^{(k)}$) for the nominal input $u_{c,t}$; set
\begin{equation}
  y_{a,t} = -(\hat y_t - \bar y),
  \label{eq:masking}
\end{equation}
so that $\tilde y_t = (y_t-\hat y_t)+\bar y \approx \bar y$.

\noindent\textbf{Computing $\hat y_t$.}
With $w_{\mathrm{ini},t} := \col(w_{t-T_{\mathrm{ini}}},\ldots,w_{t-1})$
the $T_{\mathrm{ini}}$-step window preceding $t$, split
\begin{equation}
  H^{(k)} =
  \begin{bmatrix} H_{\mathrm{p}}^{(k)} \\ H_{\mathrm{f}}^{(k)} \end{bmatrix},
  \quad
  H_{\mathrm{p}}^{(k)}\in\R^{qT_{\mathrm{ini}}\times N_k},
  \quad
  H_{\mathrm{f}}^{(k)}\in\R^{qN\times N_k}.
  \label{eq:hp-hf-phase2}
\end{equation}
With minimum-norm coefficient
$g_t := H_{\mathrm{p}}^{(k)\dagger}\, w_{\mathrm{ini},t}$, $\hat y_t$
is the output channel of the first block of $H_{\mathrm{f}}^{(k)} g_t$.
This prediction is formed \emph{before} $y_t$ is observed, enabling
pre-measurement masking.

Early in Phase~2, when $H^{(k)}$ is rank-deficient, the prediction
error may be large and the masked residual is not guaranteed to satisfy
the stealth bound; as rank grows toward $n+mL$ the prediction error
tends to decrease and masking becomes more  effective.
Throughout this phase we therefore view stealth as a numerically supported property that depends on the realized prediction error, rather than as a worst-case analytical guarantee.
\begin{remark}[A virtuous cycle]
\label{rem:virtuous-cycle}
Masking quality and admissible injection size reinforce each other.
As $H^{(k)}$ improves, the prediction error
$\norm{y_{\mathrm{nom},t}-\hat y_t}_\infty$ shrinks, so a given
injection leaves a smaller masked residual; equivalently, a larger
injection amplitude $\delta_k$ still fits within the stealth budget.
Phase~2 can therefore probe more aggressively as it proceeds.
\end{remark}

\subsection{Nested Behavioral Subspaces and Projection Error}

Denote $\BS{k} := \image H^{(k)} \subseteq \R^{qL}$, the
\emph{$k$-th identified behavioral subspace}.
When Proposition~\ref{prop:rank-increase} applies,
\begin{equation}
  \BS{k} \subsetneq \BS{k+1},
  \label{eq:nested}
\end{equation}
and each step adds exactly one new direction.

\begin{definition}[Projection-based approximation error]
\label{def:dist}
For a desired trajectory $w^\star \in \BL$ and a Hankel matrix
$H^{(k)}$ with $N_k$ columns,
\begin{equation}
  d_k(w^\star)
  := \min_{g\in\R^{N_k}} \norm{H^{(k)} g - w^\star}_2
   = \norm{w^\star - P_{\BS{k}} w^\star}_2,
  \label{eq:dist-def}
\end{equation}
the Euclidean distance from $w^\star$ to $\BS{k}$, where $P_{\BS{k}}$
is the orthogonal projector onto $\BS{k}$.
\end{definition}

\subsection{Dual Interpretation and Stopping Criteria}

Phases~1--2 admit two equivalent readings: as \emph{identification}
(each step sharpens the attacker's representation of $\BL$), and as
\emph{capability expansion} (each step enlarges the set of trajectories
the attacker can reproduce with zero prediction error).
The stopping rule trades these off against stealth cost.

\paragraph*{TC1 (Full behavioral rank)}
The attacker need not know $n$: each admissible step adds exactly one
dimension (Proposition~\ref{prop:rank-increase}) while every realized
window lies in $\BL$, capping the rank at $n+mL$, so it simply injects
until a step no longer raises the rank (\emph{rank saturation}). This
occurs exactly at $\rank H^{(k)}=n+mL$, i.e.\ $\BS{k}=\BL$, after
$T_2=T^\star-L=n+mL-1$ steps (Remark~\ref{rmk:closed-loop}); then
$d_k(w^\star)=0$ for every $w^\star\in\BL$, and the order follows as a
by-product, $n=\rank H^{(K)}-mL$.

\paragraph*{TC2 (Projection-error tolerance)}
For $\varepsilon>0$, stop at
$k_{\mathrm{stop}} := \min\{k \mid d_k(w^\star)\le\varepsilon\}$.
The attacker then either executes the \emph{projected} target
$\hat w^\star := P_{\BS{k_{\mathrm{stop}}}} w^\star$ (exactly
representable, but slightly off the intended trajectory) or the
\emph{exact} target $w^\star$, accepting a prediction error of
magnitude $\varepsilon$ that the masking must absorb as a residual.
Thus, TC1 corresponds to an exact-predictive-masking regime, whereas TC2 retains a potentially
nonzero masking residual that must satisfy the IDS constraint.

\section{Theoretical Analysis}
\label{sec:theory}

\subsection{Monotone Decrease of Projection Error}

\begin{assumption}[Arbitrary desired trajectory]
\label{ass:generic}
The desired trajectory $w^\star\in\BL$ is drawn from a measure $\mu$
on $\BL$ that is absolutely continuous with respect to the
$(n+mL)$-dimensional Lebesgue measure on $\BL$.
(Equivalently, $w^\star$ avoids every fixed lower-dimensional subspace
of $\BL$.)
The measure is on $\BL$, not $\R^{qL}$, because an executable target
must itself be a plant trajectory. This is also what makes
$d_K(w^\star)=0$ at full rank in
Corollary~\ref{cor:finite-reconstruction}.
\end{assumption}

\begin{theorem}[Almost-sure strict decrease]
\label{thm:strict-decrease}
Under Assumption~\ref{ass:generic} and $S_k \subsetneq S_{k+1}\subseteq\BL$,
$d_{k+1}(w^\star) < d_k(w^\star)$ $\mu$-almost surely.
\end{theorem}

\begin{proof}
Let $\{S_k\}_{k\ge 0}$ be subspaces of $\BL$ with
 $S_k \subseteq S_{k+1}$ and $r_k:=w^\star-P_{S_k}w^\star$.
 
Because each step adds exactly one direction
(Proposition~\ref{prop:rank-increase}),
$\dim(S_{k+1}\ominus S_k)=1$; let $v_k$ be a unit vector spanning it.
Then $P_{S_{k+1}}w^\star = P_{S_k}w^\star + \langle r_k,v_k\rangle v_k$,
so $d_{k+1}^2(w^\star) = d_k^2(w^\star) - \langle r_k,v_k\rangle^2$.
Strict decrease fails iff
$\langle (I-P_{S_k})w^\star, v_k\rangle = 0$, i.e.\ iff $w^\star$ lies
in the set
$H_k := \{w\in\BL \mid \langle(I-P_{S_k})w,v_k\rangle = 0\}$.
Since $v_k\in S_{k+1}\subseteq\BL$ is nonzero and orthogonal to $S_k$,
$H_k$ is a proper subspace of $\BL$, hence of $(n+mL)$-dimensional
Lebesgue measure zero in $\BL$; absolute continuity gives
$\mu(H_k)=0$.
\end{proof}

\begin{corollary}[Finite exact reconstruction]
\label{cor:finite-reconstruction}
Run Phase~2 to TC1, with $K=T_2$ rank-increasing steps.
Under Assumption~\ref{ass:generic},
\begin{equation}
  d_0(w^\star) > d_1(w^\star) > \cdots > d_{K}(w^\star) = 0
  \quad \mu\text{-a.s.}
\end{equation}
\end{corollary}

Corollary~\ref{cor:finite-reconstruction} formalizes the incremental
learning interpretation: the attacker's representational subspace
expands strictly until it spans $\BL$, at which point the desired
attack trajectory is exactly representable.

\begin{remark}[Excitation seed and stealth efficiency]
\label{rem:seed}
The first informative column persists in $\BS{k}$ for all later $k$
by~\eqref{eq:nested}; since the closed loop feeds it back through the
plant--controller dynamics, later rank increases may partly draw on the
system's own response, trimming the injection and stealth budget consumed.
\end{remark}

\section{Phase~3: Behavioral Covert Misappropriation}
\label{sec:phase3}

\subsection{Behavioral Predictor Construction}

After Phase~2 the attacker holds $H^{(K)}$ with $\BS{K}=\BL$ (TC1) or,
under TC2, a strict subspace containing the projected target.
Partition as in~\eqref{eq:hp-hf-phase2} with $L=T_\mathrm{ini}+N$.
Given the recent past
$w_{\mathrm{ini}} := \col(w_{t-T_\mathrm{ini}},\ldots,w_{t-1})$,
behavioral consistency yields $g\in\R^{N_K}$ with
$H_\mathrm{p}^{(K)} g = w_{\mathrm{ini}}$ (least squares under TC2), and
the predicted future is $w_\mathrm{f} = H_\mathrm{f}^{(K)} g$.

\subsection{DeePC-Style Attack Optimization}

Let $y^\star\in\R^{pN}$ be the malicious output reference and
$u^\star\in\R^{mN}$ a preferred input profile.
Split $w_\mathrm{f} = \col(u_\mathrm{f},y_\mathrm{f})$.
Phase~3 solves the convex QP
\begin{subequations}
\label{eq:deepc-attack}
\begin{align}
  \min_{g,\,u_\mathrm{f},\,y_\mathrm{f}}\;
  &\norm{y_\mathrm{f}-y^\star}_2^2
   + \lambda_u\norm{u_\mathrm{f}-u^\star}_2^2
   + \lambda_g\norm{g}_2^2
  \label{eq:deepc-obj}\\
  \text{s.t.}\;
  &\begin{bmatrix} H_\mathrm{p}^{(K)}\\ H_\mathrm{f}^{(K)} \end{bmatrix}g
   =
   \begin{bmatrix} w_{\mathrm{ini}}\\ w_\mathrm{f} \end{bmatrix},
  \label{eq:deepc-split}\\
  &\norm{\tilde y_{t+i} - \bar y}_\infty \le \varepsilon_\mathrm{ids},
   \quad i=0,\dots,N-1,
  \label{eq:deepc-ids}
\end{align}

\end{subequations}
where $\tilde y_{t+i} = y_{\mathrm{f},i} + y_{a,t+i}$, with $y_{a,t+i}$
from~\eqref{eq:masking} computed using the full-rank predictor
$H^{(K)}$, and $\lambda_u,\lambda_g\ge 0$ are regularization weights.
Constraint~\eqref{eq:deepc-ids} is linear in $(y_\mathrm{f},g)$ after
substituting the masking.
In Phase~3 the masking predictor of~\eqref{eq:masking} is evaluated for
the \emph{realized} attack input $u_{\mathrm{f}}$ rather than the
nominal $u_{c,t}$, so that $\hat y_t$ predicts the actual output under
attack and the identity $\tilde y_t\approx\bar y$ continues to hold once
the optimized injection is applied.
Problem~\eqref{eq:deepc-attack} is solved in receding-horizon fashion:
apply the first optimal input, update $w_{\mathrm{ini}}$, and re-solve.

\section{Numerical Illustration}
\label{sec:sims}

\subsection{Setup and Parameters}

We simulate a second-order discrete-time DC motor ($n=2$, $m=p=1$,
$q=2$) in closed loop with a PI controller.
The plant matrices are
$A=\left[\begin{smallmatrix}1.5462&-0.5646\\1&0\end{smallmatrix}\right]$,
$B=\left[\begin{smallmatrix}1\\0\end{smallmatrix}\right]$,
$C=\left[\begin{smallmatrix}0.3379&0.2793\end{smallmatrix}\right]$; the controller is
$A_c=\left[\begin{smallmatrix}1&-0.1673\\0&0\end{smallmatrix}\right]$,
$B_c=\left[\begin{smallmatrix}0.1701\\1\end{smallmatrix}\right]$,
$C_c=\left[\begin{smallmatrix}1&-0.1673\end{smallmatrix}\right]$,
$D_c=0.1701$. A 200-step warm-up under the nominal setpoint $r=1.0$ drives the loop
to steady state, from which $(\bar u,\bar y)$ and the bounds
$\varepsilon_{\mathrm{ss}}$, $\hat G_\infty$ are estimated.
Note that the plant and controller matrices above are used only to
generate the simulation data. The attacker has no access to them. We set $T_{\mathrm{ini}}=2$ (equal to the plant order $n$, the minimum
admissible value) and $N=20$, giving Hankel depth
$L=T_{\mathrm{ini}}+N=22$ and full behavioral rank $n+mL=24$.

\textbf{Phase~1}: $T_1=L=22$ samples; Rademacher probing
    with amplitude $\delta_1=10^{-3}$, satisfying the stealth budget
    $\delta_1\le(\varepsilon_{\mathrm{ids}}-\varepsilon_{\mathrm{ss}})
    /\hat G_\infty$; no masking applied.
    Yields $\rank H^{(0)}=1$. \textbf{Phase~2}: rank-increasing injections
    (Proposition~\ref{prop:rank-increase}) with amplitude
    $\delta_k=10^{-2}$ and masking~\eqref{eq:masking}; terminates at TC1
    ($\rank H^{(K)}=n+mL=24$) after $T_2=n+mL-1=23$ steps, total length
    $T_1+T_2=T^\star=45$. \textbf{Phase~3}: 100-step receding-horizon
    QP~\eqref{eq:deepc-attack}, $\lambda_g=10^{-5}$, $\lambda_u=10^{-3}$,
    $u\in[-2,2]$, $\varepsilon_{\mathrm{ids}} = 4\,\delta_1\hat{G}_\infty$,
    tracking the sinusoid $y^\star(t) = \bar y + A\sin(2\pi t/P)$,
    $A=0.5$,
    $P=50$\footnote{Code available at \url{https://github.com/v-mattos/covert-behavioral-id/}}.

\subsection{Results}

\definecolor{plotblue}{RGB}{0,114,178}
\definecolor{plotorange}{RGB}{230,159,0}
\definecolor{plotvermillion}{RGB}{213,94,0}

\pgfplotstableread{%
t residual
1 0.002656
2 0.003073
3 0.003277
4 0.002780
5 0.002549
6 0.003059
7 0.004022
8 0.004540
9 0.004167
10 0.003906
11 0.004275
12 0.004356
13 0.004388
14 0.004259
15 0.004160
16 0.003974
17 0.003878
18 0.004415
19 0.004645
20 0.004796
21 0.005425
22 0.005624
23 0.000920
24 0.000616
25 0.000449
26 0.000100
145 0.000100
}\figIIMasked

\pgfplotstableread{%
t residual
1 0.002656
2 0.003073
3 0.003277
4 0.002780
5 0.002549
6 0.003059
7 0.004022
8 0.004540
9 0.004167
10 0.003906
11 0.004275
12 0.004356
13 0.004388
14 0.004259
15 0.004160
16 0.003974
17 0.003878
18 0.004415
19 0.004645
20 0.004796
21 0.005425
22 0.005624
23 0.004979
24 0.004172
25 0.003938
26 0.004336
27 0.005074
28 0.005972
29 0.006937
30 0.007935
31 0.008947
32 0.009963
33 0.010976
34 0.011982
35 0.012980
36 0.013969
37 0.014948
38 0.015918
39 0.016878
40 0.017829
41 0.018771
42 0.019705
43 0.020631
44 0.021549
45 0.022459
46 0.023362
47 0.045243
48 0.092242
49 0.144371
50 0.191739
51 0.236391
52 0.282366
53 0.327993
54 0.369277
55 0.404255
56 0.433180
57 0.456269
58 0.472827
59 0.482008
60 0.483521
61 0.477511
62 0.464163
63 0.443600
64 0.416048
65 0.381930
66 0.341818
67 0.296353
68 0.246233
69 0.192232
70 0.135203
71 0.076050
72 0.015708
73 0.044876
74 0.104750
75 0.162968
76 0.218612
77 0.270803
78 0.318720
79 0.361608
80 0.398789
81 0.429677
82 0.453786
83 0.470735
84 0.480256
85 0.482201
86 0.476537
87 0.463355
88 0.442862
89 0.415381
90 0.381346
91 0.341294
92 0.295856
93 0.245749
94 0.191763
95 0.134749
96 0.075607
97 0.015269
98 0.045312
99 0.105183
100 0.163398
101 0.219040
102 0.271231
103 0.319148
104 0.362034
105 0.399215
106 0.430103
107 0.454212
108 0.471160
109 0.480682
110 0.482626
111 0.476962
112 0.463780
113 0.443287
114 0.415806
115 0.381772
116 0.341719
117 0.296282
118 0.246174
119 0.192188
120 0.135175
121 0.076033
122 0.015695
123 0.044887
124 0.104758
125 0.162973
126 0.218615
127 0.270806
128 0.318722
129 0.361609
130 0.398790
131 0.429678
132 0.453786
133 0.470735
134 0.480257
135 0.482201
136 0.476537
137 0.463355
138 0.442862
139 0.415381
140 0.381346
141 0.341294
142 0.295856
143 0.245749
144 0.191763
145 0.134749
}\figIIUnmasked

\pgfplotstableread{%
t y
1 0.994650
2 0.995067
3 0.995271
4 0.994774
5 0.994543
6 0.995053
7 0.996017
8 0.996534
9 0.996161
10 0.995900
11 0.996269
12 0.996350
13 0.996382
14 0.996253
15 0.996154
16 0.995968
17 0.995872
18 0.996410
19 0.996640
20 0.996790
21 0.997419
22 0.997618
23 0.996973
24 0.996166
25 0.995932
26 0.996330
27 0.997069
28 0.997966
29 0.998931
30 0.999929
31 1.000941
32 1.001957
33 1.002970
34 1.003976
35 1.004974
36 1.005963
37 1.006942
38 1.007912
39 1.008872
40 1.009823
41 1.010766
42 1.011699
43 1.012625
44 1.013543
45 1.014453
46 1.015357
47 1.037237
48 1.084236
49 1.136365
50 1.183733
51 1.228385
52 1.274360
53 1.319987
54 1.361271
55 1.396250
56 1.425175
57 1.448263
58 1.464822
59 1.474002
60 1.475515
61 1.469506
62 1.456157
63 1.435594
64 1.408042
65 1.373925
66 1.333813
67 1.288347
68 1.238227
69 1.184226
70 1.127197
71 1.068045
72 1.007702
73 0.947118
74 0.887244
75 0.829026
76 0.773383
77 0.721191
78 0.673274
79 0.630387
80 0.593206
81 0.562317
82 0.538208
83 0.521260
84 0.511738
85 0.509794
86 0.515457
87 0.528640
88 0.549133
89 0.576613
90 0.610648
91 0.650700
92 0.696138
93 0.746245
94 0.800231
95 0.857245
96 0.916387
97 0.976725
98 1.037307
99 1.097178
100 1.155393
101 1.211035
102 1.263225
103 1.311142
104 1.354029
105 1.391209
106 1.422097
107 1.446206
108 1.463155
109 1.472676
110 1.474620
111 1.468957
112 1.455774
113 1.435281
114 1.407801
115 1.373766
116 1.333714
117 1.288276
118 1.238169
119 1.184183
120 1.127169
121 1.068027
122 1.007689
123 0.947107
124 0.887236
125 0.829021
126 0.773379
127 0.721189
128 0.673272
129 0.630385
130 0.593205
131 0.562316
132 0.538208
133 0.521259
134 0.511738
135 0.509794
136 0.515457
137 0.528639
138 0.549133
139 0.576613
140 0.610648
141 0.650700
142 0.696138
143 0.746245
144 0.800231
145 0.857245
}\figIIYtrue

\pgfplotstableread{%
t y
46 0.991994
47 1.054661
48 1.116339
49 1.176057
50 1.232871
51 1.285887
52 1.334268
53 1.377251
54 1.414158
55 1.444408
56 1.467523
57 1.483138
58 1.491008
59 1.491008
60 1.483138
61 1.467523
62 1.444408
63 1.414158
64 1.377251
65 1.334268
66 1.285887
67 1.232871
68 1.176057
69 1.116339
70 1.054661
71 0.991994
72 0.929328
73 0.867649
74 0.807932
75 0.751117
76 0.698102
77 0.649721
78 0.606738
79 0.569830
80 0.539581
81 0.516466
82 0.500851
83 0.492981
84 0.492981
85 0.500851
86 0.516466
87 0.539581
88 0.569830
89 0.606738
90 0.649721
91 0.698102
92 0.751117
93 0.807932
94 0.867649
95 0.929328
96 0.991994
97 1.054661
98 1.116339
99 1.176057
100 1.232871
101 1.285887
102 1.334268
103 1.377251
104 1.414158
105 1.444408
106 1.467523
107 1.483138
108 1.491008
109 1.491008
110 1.483138
111 1.467523
112 1.444408
113 1.414158
114 1.377251
115 1.334268
116 1.285887
117 1.232871
118 1.176057
119 1.116339
120 1.054661
121 0.991994
122 0.929328
123 0.867649
124 0.807932
125 0.751117
126 0.698102
127 0.649721
128 0.606738
129 0.569830
130 0.539581
131 0.516466
132 0.500851
133 0.492981
134 0.492981
135 0.500851
136 0.516466
137 0.539581
138 0.569830
139 0.606738
140 0.649721
141 0.698102
142 0.751117
143 0.807932
144 0.867649
145 0.929328
}\figIIYref

\begin{figure}[t]
\centering
\begin{tikzpicture}
\begin{groupplot}[
    group style={
        group size=1 by 3,
        vertical sep=0.8cm,
    },
    width=\columnwidth,
    tick label style={font=\footnotesize},
    label style={font=\normalsize},
    grid=major,
    grid style={black!25},
    axis line style={black}
]

\nextgroupplot[
    height=0.45\columnwidth,
    xmin=-5, xmax=120,
    xtick={0,30,60,90,120},
    xticklabels={,,},
    ylabel={Residual},
    ymode=log,
    ymin=1e-5, ymax=100,
    ytick={1e-5,1e-4,1e-3,1e-2,1e-1,1e0},
    legend style={
        font=\footnotesize,
        draw=black,
        fill=white,
        cells={anchor=west},
        inner sep=1.5pt,
        at={(0.5,1.1)},
        anchor=south,
        legend columns=3
    }
]

\addplot[plotorange, thick, mark=none]
    table[x=t,y=residual] {\figIIUnmasked};
\addlegendentry{$\lVert y_t-\bar y\rVert_\infty$ (unmasked)}

\addplot[plotblue, thick, mark=none]
    table[x=t,y=residual] {\figIIMasked};

\addlegendentry{$\lVert \tilde y_t-\bar y\rVert_\infty$ (masked)}

\addplot[red!50, dashed, thick]
    coordinates {(-3,0.021827) (149,0.021827)};
\addlegendentry{$\varepsilon_{\mathrm{ids}}$}

\addplot[black!45, densely dashed, very thick, forget plot]
    coordinates {(22.5,1e-6) (22.5,100)};
\addplot[black!45, densely dashed, very thick, forget plot]
    coordinates {(45.5,1e-6) (45.5,100)};

\node[font=\footnotesize\bfseries,text=black,fill=white,fill opacity=1,
     text opacity=1,draw=black!60,line width=0.55pt,rounded corners=2pt,
     inner xsep=4pt,inner ysep=2pt,align=center]
  at (axis cs:9,10) {Phase\\[-1pt]1};
\node[font=\footnotesize\bfseries,text=black,fill=white,fill opacity=1,
     text opacity=1,draw=black!60,line width=0.55pt,rounded corners=2pt,
     inner xsep=4pt,inner ysep=2pt,align=center]
  at (axis cs:34,10) {Phase\\[-1pt]2};
\node[font=\footnotesize\bfseries,text=black,fill=white,fill opacity=1,
     text opacity=1,draw=black!60,line width=0.55pt,rounded corners=2pt,
     inner xsep=4pt,inner ysep=2pt,align=center]
  at (axis cs:69,10) {Phase\\[-1pt]3};

\nextgroupplot[
    height=0.34\columnwidth,
    xlabel={Time step},
    xlabel style={yshift=0.15cm},
    xmin=-5, xmax=120,
    xtick={0,30,60,90,120},
    ylabel={$y_t$},
    ymin=0.4, ymax=1.55,
    ytick={0.6,1.0,1.4}, yticklabel style={ /pgf/number format/fixed, /pgf/number format/precision=1, /pgf/number format/fixed zerofill },
    legend style={
        font=\scriptsize,
        draw=black,
        fill=white,
        cells={anchor=west},
        inner sep=1.2pt,
        at={(0.5,1.13)},
        anchor=south,
        legend columns=3
    }
]
\addplot[plotblue, thick]
    table[x=t,y=y] {\figIIYtrue};
\addlegendentry{$y_t$}

\addplot[plotvermillion, dash dot, thick]
    table[x=t,y=y] {\figIIYref};
\addlegendentry{$y^\star(t)$}

\addplot[black, dotted, thick]
    coordinates {(-3,0.991994) (149,0.991994)};
\addlegendentry{$\bar y$}
 \addplot[black!45, densely dashed, very thick, forget plot]
    coordinates {(22.5,-0.05) (22.5,1.7)};
\addplot[black!45, densely dashed, very thick, forget plot]
    coordinates {(45.5,-0.05) (45.5,1.7)};

\nextgroupplot[
    name=panelc,
    height=0.4\columnwidth,
    xlabel={Phase 2 iteration $k$},
    ylabel={$d_k(w^\star)/d_0(w^\star)$},
    tick label style={font=\scriptsize},
    grid=both,
    major grid style={gray!45},
    minor grid style={gray!20},
    xmin=-0.5, xmax=23.5,
    ymin=-0.03, ymax=1.05,
    xtick={0,10,20},
    ytick={0,0.25,0.50,0.75,1.00},
    yticklabel style={
        /pgf/number format/fixed,
        /pgf/number format/precision=2
    },
]
\addplot[
    blue,
    thick,
    mark=*,
    mark size=0.8pt
] coordinates {
    (0,1.00000)
    (1,0.99982)
    (2,0.99953)
    (3,0.99885)
    (4,0.99231)
    (5,0.97000)
    (6,0.94297)
    (7,0.93256)
    (8,0.93218)
    (9,0.92942)
    (10,0.92687)
    (11,0.91431)
    (12,0.87005)
    (13,0.82322)
    (14,0.77455)
    (15,0.72574)
    (16,0.64643)
    (17,0.56159)
    (18,0.49378)
    (19,0.45362)
    (20,0.24130)
    (21,0.06364)
    (22,0.01367)
    (23,0.00000)
};

\end{groupplot}
\end{tikzpicture}
\caption{Attack timeline (Phases~1--3) and Phase~2 identification.
\emph{Top}: masked residual $\lVert\tilde y_t-\bar y\rVert_\infty$
(blue) and unmasked residual $\lVert y_t-\bar y\rVert_\infty$ (orange)
vs.\ threshold~$\varepsilon_{\mathrm{ids}}$ (red, dashed); log axis,
both traces clipped at a display floor of~$10^{-4}$. \emph{Middle}:
output~$y_t$ tracking reference~$y^\star(t)$ about~$\bar y$; dashed
lines mark phase boundaries. The masked residual is $\approx 0$ here,
so the IDS constraint~\eqref{eq:deepc-ids} binds only once a
prediction error is present. \emph{Bottom}: normalized projection
error $d_k(w^\star)/d_0(w^\star)$, decreasing to zero at full
behavioral rank $n+mL=24$ after $T_2=23$ steps.}
\label{fig:unified}
\end{figure}
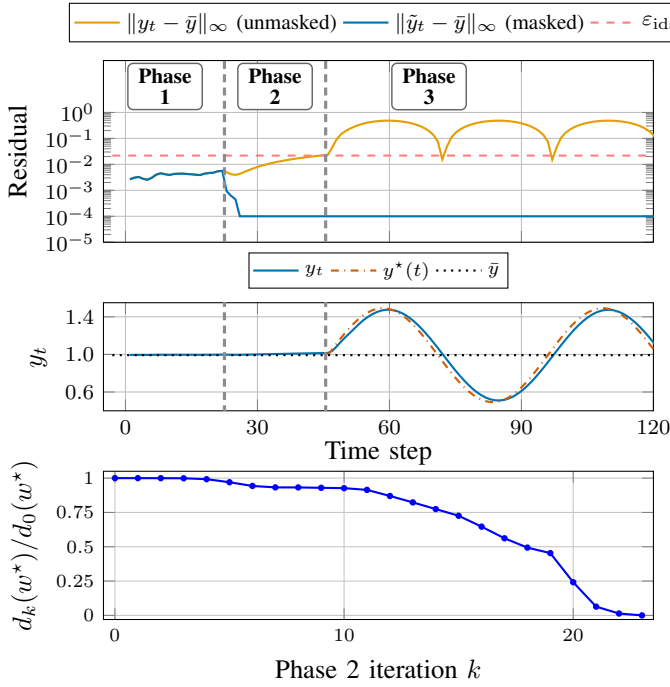

Figure~\ref{fig:unified} presents the complete attack timeline.
During Phase~1 ($t=0$--$21$) the low-amplitude injection produces
negligible deviation, with masked and unmasked residuals
indistinguishable and well below $\varepsilon_{\mathrm{ids}}$. In
Phase~2 ($t=22$--$44$) the unmasked residual rises as rank-increasing
injections grow more informative, while the masked residual stays near
zero, illustrating the virtuous cycle of
Remark~\ref{rem:virtuous-cycle}. Once Phase~3 begins ($t\ge 45$) the
plant output tracks the sinusoidal reference $y^\star(t)$
(oscillating $\pm0.5$ about $\bar y$) while the masked signal stays
flat; the unmasked residual exceeds $\varepsilon_{\mathrm{ids}}$ at
every peak, confirming the attack would have been detected without
masking, and the masked residual never triggers the alarm. The bottom
panel plots the normalized projection error
$d_k(w^\star)/d_0(w^\star)$ across the $T_2=23$ Phase-2 iterations: it
decreases monotonically to zero as the Hankel rank climbs from $1$ to
$n+mL=24$, confirming Corollary~\ref{cor:finite-reconstruction}: the
attacker's representational subspace expands at each rank increase
until the desired trajectory is exactly representable.

 \section{Discussion and Limitations}
\label{sec:discussion}

All analysis assumes noise-free LTI dynamics. Under measurement or
process noise, the rank test behind Proposition~\ref{prop:rank-increase}
becomes a thresholded numerical-rank decision: the strict
nesting~\eqref{eq:nested} holds only approximately, $d_k(w^\star)$
floors at a noise-dependent level instead of reaching zero, weakening
Corollary~\ref{cor:finite-reconstruction} to approximate
reconstruction, and stealth becomes probabilistic. A robust treatment,
via a quantitative persistency-of-excitation notion and the robust
fundamental lemma~\cite{Coulson2022robust}, remains an open direction.
The infinity-norm IDS itself is analytically tractable but weaker than
CUSUM or $\chi^2$ tests; adapting the masking conditions to statistical
detectors would require probabilistic arguments beyond the present scope.
Phase~2 costs $O((qL)^2)$ per step; Phase~3 solves the convex
QP~\eqref{eq:deepc-attack}, scaling with $N_K$ and $N$. Both suit
offline or slow-sampled SCADA use; real-time settings may need
warm-starting.
The attack's own regularity, systematic rank growth and
predictor-consistent residuals, is itself a detectable signature that
statistical tests for sudden PE or correlation-structure changes could
flag. The Phase~2 update resembles incremental subspace tracking, suggesting extension to slowly
time-varying plants via periodic Hankel updates.

\section{Conclusion}
\label{sec:conc}

We presented a three-phase, model-free behavioral framework for covert
misappropriation under a MITM threat model. Phase~1 seeds a
rank-deficient Hankel matrix via low-amplitude Rademacher probing.
Phase~2 increases the Hankel rank under stealth constraints; the
attacker's additive actuator authority makes this closed-loop
procedure equivalent to the open-loop design of
\cite{waardePersistentExcitationOnline2021}, reaching full behavioral
rank in $T^\star=n+(m+1)L-1$ samples and recovering the plant order as
a by-product. The identified subspaces are strictly nested, and the
projection error of any desired trajectory decreases monotonically,
almost surely strictly at each rank increase, reaching zero at full
rank (Theorem~\ref{thm:strict-decrease},
Corollary~\ref{cor:finite-reconstruction}). Phase~3 deploys the
learned Hankel matrix in a DeePC-style convex QP under explicit IDS
constraints. These results show that DeePC can be repurposed for adversarial ends,
offering both a practical attack architecture and a rigorous
theoretical foundation. Open directions include noise robustness, richer
IDS models, sharper prediction-error bounds, and complementary
detection strategies.

\bibliographystyle{IEEEtran}
\bibliography{bibliography}

\end{document}